\documentclass[12pt]{article}

\usepackage{amssymb}
\usepackage{fullpage}
\usepackage{graphicx}

\newtheorem{theorem}{Theorem}
\newtheorem{lemma}{Lemma}

\newcommand{\IR}{\mathbb{R}}
\newcommand{\Yao}{\mathord{\it Yao}}
\newcommand{\CPD}{\mathord{\it CPD}}
\newcommand{\VD}{\mathord{\it VD}}
\newcommand{\DHP}{\mathord{\it DH}}
\newcommand{\Del}{\mathord{\it Del}_{\square}}
\newcommand{\DS}{\mathord{\it DS}}
\newcommand{\RC}{\mathord{\it RC}}
\newcommand{\MW}{\mathord{\it MW}}

\newcommand{\qed}{\rule{0.5em}{1.5ex}}
\newcommand{\fqed}{{\hfill~\qed}}
\newenvironment{proof}{{\noindent \bf Proof.}}
                      {{\hfill \fqed} \vspace{1em}}

\title{Closest-Pair Queries in Fat 
Rectangles\thanks{S.W.\ Bae was supported by the Basic Science Research 
Program through the National Research Foundation of Korea (NRF) 
funded by the Ministry of Education (2018R1D1A1B07042755). 
M.\ Smid was supported by the Natural Sciences and Engineering Research 
Council of Canada (NSERC).}}
\author{
Sang Won Bae\thanks{Division of Computer Science and Engineering, 
    Kyonggi University, Suwon, South Korea.}
\and 
Michiel Smid\thanks{School of Computer Science, Carleton University,
    Ottawa, Canada.} 
}
\date{\today}

\begin{document} 

\maketitle 

\begin{abstract} 
In the range closest pair problem, we want to construct a data structure 
storing a set $S$ of $n$ points in the plane, such that for any 
axes-parallel query rectangle $R$, the closest pair in the set $R \cap S$ 
can be reported. The currently best result for this problem is by 
Xue et al.~(SoCG 2018). Their data structure has size $O(n \log^2 n)$ and 
query time $O(\log^2 n)$. We show that a data structure of size 
$O(n \log n)$ can be constructed in $O(n \log n)$ time, such that 
queries can be answered in $O(\log n + f \log f)$ time, 
where $f$ is the aspect ratio of $R$. Thus, for fat query rectangles, 
the query time is $O(\log n)$. This result is obtained by reducing the 
range closest pair problem to standard range searching problems on 
the points of $S$. 
\end{abstract} 

\section{Introduction} 
Range searching and closest pair problems have been well-studied in 
computational geometry. In both problems, we are given a finite set $S$ 
of points. In the range searching problem, we want to construct a data 
structure, such that for any query range $R$, the elements of $R \cap S$ 
can be reported or counted. In the closest pair problem, the goal is to 
design an algorithm that computes the closest pair in the set $S$. 
Overviews of the main results for these problems can be found in the 
survey papers by Agarwal and Erickson~\cite{ae-99} and Smid~\cite{s-97}.    

In this paper, we consider the \emph{range closest pair} problem: 
Given a set $S$ of $n$ points in the plane, construct a data structure 
such that for any axes-parallel query rectangle $R$, the closest pair in 
the set $R \cap S$ can be reported.   

The range closest pair problem was introduced by 
Shan et al.~\cite{szs-03}. They presented experimental results based 
on R-trees, but did not give a complexity analysis. 
Gupta~\cite{g-05} showed that, for a set of $n$ points in one-dimensional 
space, a query can be answered in $O(\log n)$ time using a data structure 
of size $O(n)$. For the two-dimensional problem, he presented a data 
structure of size $O(n^2 \log^3 n)$ that has $O(\log^3 n)$ query time. 
Sharathkumar and Gupta~\cite{sg-07} improved the space bound to 
$O(n \log^3 n)$, while keeping the query time at $O(\log^3 n)$. 
They first prove this result for \emph{fat} query rectangles, i.e., 
rectangles whose aspect ratio (which is the longest side length divided 
by the shortest side length) is bounded from above by a constant. 
Then, they prove the same result for arbitrary query rectangles. 
Gupta et al.~\cite{gjks-14} gave a data structure of size $O(n \log^5 n)$,
that supports queries in $O(\log^2 n)$ time. Abam et al.~\cite{acfs-13} 
presented a variant of this data structure: The space and query time are 
the same as in~\cite{gjks-14}, but their structure can be constructed in 
$O(n \log^5 n)$ time. The currently best result is by 
Xue et al.~\cite{xlrj-18}. They presented a data structure of size   
$O(n \log^2 n)$ that supports queries in $O(\log^2 n)$ time. It is 
not known if their data structure can be constructed efficiently, say 
in subquadratic time.  

Xue et al.~\cite{xlj-18} introduced an approximate version of the range 
closest pair problem and gave a solution for the $d$-dimensional case, 
for any constant dimension $d \geq 1$. The results in~\cite{xlj-18} 
imply that a data structure of size $O(n \log^{d-1} n)$ can be 
constructed in $O(n \log^{d-1} n)$ time, such that the following holds: 
For any axes-parallel query rectangle $R$ in $\IR^d$ and any query 
value $\epsilon>0$, a pair $p,q$ of distinct points in $S$ can be 
computed, such that (i) $p \in R$, (ii) $q$ is contained in the 
expanded rectangle obtained by scaling $R$ by a factor of $1+\epsilon$ 
with respect to its center, and (iii) the distance between $p$ and $q$ 
is at most the closest-pair distance in $R \cap S$. Such a pair $p,q$ 
can be computed in $O(\log^{d-1} n + (f/\epsilon)^d \log (f/\epsilon))$ 
time, where $f$ denotes the aspect ratio of the query rectangle $R$. 
Observe that, since the point $q$ can be outside of $R$, the distance 
between $p$ and $q$ may be much smaller than the closest-pair distance 
in $R \cap S$. 

\subsection{Our Result} 
\label{secOR} 
We show that the approach of Sharathkumar and Gupta~\cite{sg-07} for fat 
query rectangles can be improved. Our main result is the following: 

\begin{theorem} 
\label{thm1} 
Let $S$ be a set of $n$ points in the plane. In $O(n \log n)$ time, we 
can construct a data structure of size $O(n \log n)$, such that for any 
axes-parallel query rectangle $R$, the closest pair in $R \cap S$ can 
be computed in $O(\log n + f \log f)$ time, where $f$ is the aspect 
ratio of $R$.
\end{theorem}

Thus, for fat query rectangles, i.e., rectangles whose aspect ratio is 
bounded by a constant, the query time is $O(\log n)$. In fact, the 
query time is $O(\log n)$ for any rectangle having aspect ratio 
$O(\log n / \log\log n)$.

Our result will be based on the following two claims. (These claims are
stronger variants of similar claims in Sections~4.1 and~4.2 
of~\cite{sg-07}.) Let $R$ be a query rectangle, let $\ell$ be its 
shortest side length, and let $f$ be its aspect ratio. 
\begin{itemize} 
\item If $R$ contains $O(f)$ points of $S$, then we can answer a query 
by using any optimal closest pair algorithm. 
\item Assume that $R$ contains $\Omega(f)$ points of $S$. 
A standard packing argument implies that the closest-pair distance in 
$R \cap S$ is less than $\ell/2$. In this case, we can compute four small 
squares ``anchored'' at the four vertices of $R$, each one containing 
$O(1)$ points of~$S$. That is, for each vertex $v$ of $R$, one of these 
small squares has $v$ as a vertex and is contained in $R$. If $p,q$ is 
the closest pair in $R \cap S$, then either (i) both $p$ and $q$ are 
contained in the same small square or (ii) at least one of $(p,q)$ and 
$(q,p)$ is an edge in the Yao-graph that uses four cones of angle 
$\pi/2$ (this graph will be defined in Section~\ref{secprelim}). 
This claim will be proved in Section~\ref{seccorrect}. 
\end{itemize}  
Based on these two claims, a range closest pair query will be reduced to 
several standard range searching queries on two-dimensional point sets. 
Note that in~\cite{sg-07}, the query is reduced to range searching queries 
in four-dimensional space. We are able to reduce the dimension from four 
to two, because we use a different way to partition the query rectangle 
into four small anchored squares and five other rectangles. 

The rest of this paper is organized as follows. In 
Section~\ref{secprelim}, we recall the basic packing argument, some basic 
range searching problems, and the Yao-graph. Our range closest pair data 
structure and the query algorithm will be given in Section~\ref{secmain}. 
In Section~\ref{seccon}, we give a simple example that shows that a 
point set $S$ may contain $\Omega(n^2)$ pairs of points, each of which 
is the closest pair for some fat rectangle. If we restrict queries to 
squares, however, then the number of possible closest pairs is only 
$O(n)$. We prove this claim by showing that each such pair is an edge 
in the second-order $L_{\infty}$-Delaunay graph.

\section{Preliminaries} 
\label{secprelim} 

For any two points $p$ and $q$ in the plane, we denote their Euclidean 
distance by $|pq|$. For any finite set $S$ of points in the plane, a 
\emph{closest pair} in $S$ is a pair $p,q$ of distinct points in $S$ 
whose distance $|pq|$ is minimum. The \emph{closest-pair distance}, i.e., 
the distance $|pq|$ of a closest pair $p,q$, will be denoted by 
$\CPD(S)$. If $S$ contains at most one point, then $\CPD(S) = \infty$. 

A proof of the following result can be found in the textbook by 
Cormen et al.~\cite{clrs-09}. 

\begin{lemma} 
\label{lemma3} 
Let $S$ be a set of $n$ points in the plane. The closest pair in $S$ can 
be computed in $O(n \log n)$ time.  
\end{lemma} 

Throughout the rest of this paper, a \emph{rectangle} refers to a 
region $R$ in the plane defined by a Cartesian product 
$R = [a_x,b_x] \times [a_y,b_y]$, where $a_x$, $a_y$, $b_x$, and $b_y$ 
are real numbers with $a_x < b_x$ and $a_y < b_y$. 
The \emph{aspect ratio}, or \emph{fatness}, of $R$ is the ratio of 
$\max(b_x-a_x,b_y-a_y)$ and $\min(b_x-a_x,b_y-a_y)$.

Throughout the rest of this paper, we assume that no two points in $S$ 
share a coordinate along any dimension. We also assume, for ease of 
presentation, that the $n \choose 2$ distances defined by the pairs of 
points in $S$ are distinct, so that the closest pair in $R \cap S$ is 
uniquely defined for any rectangle $R$. (If distances are not distinct, 
then we can take the closest pair in $R \cap S$ to be the 
lexicographically smallest pair that achieves the minimum distance. 
All arguments in this paper will still be valid.) 

In the following lemma, we present the standard packing argument to 
prove the claim that if a rectangle $R$ contains ``many'' points, 
then the closest-pair distance in $R$ is ``small''. 
  
\begin{lemma} 
\label{lemma1} 
Let $S$ be a finite set of points in the plane, let $R$ a rectangle, 
let $\ell$ be the shortest side length of $R$, and let $f$ be the aspect 
ratio of $R$. 
\begin{enumerate} 
\item If $|R \cap S| > 4 \lceil 4f \rceil$, then 
      $\CPD(R \cap S) < \ell/2$. 
\item If $R$ is a square and $|R \cap S| \geq 5$, then 
      $\CPD(R \cap S) < \ell$. 
\end{enumerate} 
\end{lemma} 
\begin{proof} 
To prove the first claim, we assume, without loss of generality, that 
the vertical sides of $R$ have length $\ell$ and, thus, the horizontal 
sides have length $f \ell$. Partition $R$ into four horizontal slabs, 
each having height $\ell/4$, and $\lceil 4f \rceil$ vertical slabs, each 
having width at most $\ell/4$. These slabs partition 
$R$ into $4 \lceil 4f \rceil$ rectangles, each one having sides of 
length at most $\ell/4$. By the Pigeonhole Principle, 
one of these rectangles contains at least two points of $R \cap S$. 
These two points have distance at most $\sqrt{2} \cdot \ell/4 < \ell/2$.  

The proof of the second claim is similar. In this case, we divide the 
square $R$ into four squares, each one having sides of length $\ell/2$. 
\end{proof} 

Lemmas~\ref{lemma2}--\ref{lemma4} below are based on range trees. 
Let $S$ be a set of $n$ points in the plane. A \emph{range tree} consists 
of a balanced binary tree $T$ storing the points of $S$ at its leaves, 
in sorted order of their $x$-coordinates. For each node $u$ of $T$, let
$S_u$ be the set of points of $S$ that are stored in $u$'s subtree. 
We store with $u$ a pointer to an array $A_u$ storing the points of 
$S_u$, in sorted order of their $y$-coordinates. (For a more detailed 
description, see, e.g., the textbook by de Berg et al.~\cite{bcko-08}.)  

Consider a query rectangle $R$. Using the technique of 
\emph{fractional cascading}, in $O(\log n)$ total time, a sequence 
$u_1,u_2,\ldots,u_k$ of $k=O(\log n)$ nodes in $T$ can be computed, 
together with indices $\alpha_i$ and $\beta_i$, for $i=1,2,\ldots,k$, 
such that the $k$ subarrays $A_{u_i}[\alpha_i \ldots \beta_i]$ form a 
partition of $R \cap S$. Thus, by reporting the points stored 
in these subarrays, we report each point of the set $R \cap S$ exactly 
once. Also, the size of the set $R \cap S$ is equal to 
$\sum_{i=1}^k ( \beta_i - \alpha_i + 1)$. (Again, for more details, 
refer to de Berg et al.~\cite{bcko-08}.) The following lemma summarizes 
this standard application of range trees. 

\begin{lemma} 
\label{lemma2} 
Let $S$ be a set of $n$ points in the plane. In $O(n \log n)$ time, we 
can construct a data structure of size $O(n \log n)$, such that for any 
query rectangle $R$, 
\begin{enumerate}
\item the elements of the set $R \cap S$ can be reported in 
      $O(\log n + |R \cap S|)$ time, 
\item the size of the set $R \cap S$ can be reported in $O(\log n)$ 
      time. 
\end{enumerate} 
\end{lemma} 

Instead of reporting or counting the elements of $R \cap S$, we will 
also need a data structure for the case in which the points of $S$ 
are weighted and we want to report the minimum weight of any point 
inside the query rectangle. 

\begin{lemma} 
\label{lemma23} 
Let $S$ be a set of $n$ weighted points in the plane. In $O(n \log n)$ 
time, we can construct a data structure of size $O(n \log n)$, such that 
for any query rectangle $R$, the minimum weight of any point in 
$R \cap S$ can be reported in $O(\log n )$ time.  
\end{lemma} 
\begin{proof}
Consider a range tree for $S$, as described above. For each node $u$ 
of the tree $T$, we use the algorithm of 
Bender and Farach-Colton~\cite{bf-00} to preprocess the array $A_u$ in 
$O(|A_u|)$ time, such that for any subarray, the minimum weight of any 
point in the subarray can be reported in $O(1)$ time. 

Given a query rectangle $R$, we compute, in $O(\log n)$ total time,  
the sequence $u_i$, $i = 1,2,\ldots,k=O(\log n)$, of nodes in $T$, and 
the sequences of indices $\alpha_i$ and $\beta_i$, $i=1,2,\ldots,k$. 
Then, for each $i$, we compute the minimum weight of any point in the 
subarray $A_{u_i}[\alpha_i,\beta_i]$ in $O(1)$ time. From this, we 
obtain the minimum weight of the points in $R \cap S$ in $O(\log n)$ 
time.  
\end{proof} 

The following lemma gives the data structure that we will need to 
find the anchored squares that were mentioned in Section~\ref{secOR}. 

\begin{lemma} 
\label{lemma4} 
Let $S$ be a set of $n$ points in the plane and let $c$ be a constant. 
In $O(n \log n)$ time, we can construct a data structure of size 
$O(n \log n)$, such that for any query point $q$ in~$\IR^2$, we can 
compute, in $O(\log n)$ time, the smallest square with bottom-left 
corner at $q$ that contains at least $c$ points of $S$. If such a 
square does not exist, then the query returns the infinite square 
with bottom-left corner at $q$. 
\end{lemma}
\begin{proof} 
Let $V_q$ be the vertical line through $q$, let $H_q$ be the 
horizontal line through $q$, and let $D_q$ be the line through $q$ 
that makes an angle of $\pi/4$ with the positive $x$-axis. 
Let $\VD_q$ be the cone consisting of all points in the plane that 
are to the right of $V_q$ and above $D_q$. 
Let $\DHP_q$ be the cone consisting of all points in the plane that 
are above $H_q$ and below $D_q$. (Refer to Figure~\ref{figVq}.) 

\begin{figure}
\begin{center}
\includegraphics[scale=0.65]{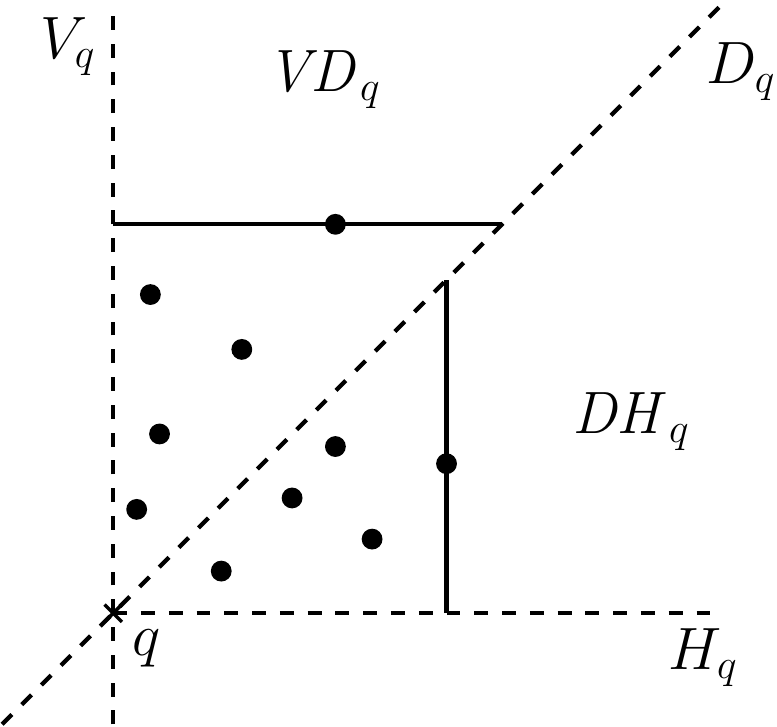}
\end{center}
\caption{Illustrating the proof of Lemma~\ref{lemma4} for $c=5$.}
\label{figVq}
\end{figure}

If we are given the $c$ lowest points of $\VD_q \cap S$ and the $c$ 
leftmost points of $\DHP_q \cap S$, then we can answer the query in 
$O(c) = O(1)$ time. Below, we show how a variant of the range tree can 
be used to compute the $c$ lowest points of $\VD_q \cap S$ in 
$O(\log n)$ time. The $c$ leftmost points of $\DHP_q \cap S$ can be 
computed in a symmetric way. 

Consider a range tree for $S$, as described before. Recall that the 
balanced binary tree $T$ stores the points of $S$ at its leaves, in 
sorted order of their $x$-coordinates. For each node $u$ of $T$, we 
store the points $p=(p_x,p_y)$ of $S_u$ in the array $A_u$. Instead of 
storing these points in sorted order of their $y$-coordinates, we store 
them in sorted order of their values $p_y-p_x$. Each point $p$ in $A_u$ 
gets the value $p_y$ as its weight. With each such point $p$, we also 
store the $c$ smallest weights in the suffix of $A_u$ that starts at $p$. 
(If this suffix has length less than $c$, then we store with $p$ all 
points in the suffix.) The additional information stored with this array 
can be computed in time that is proportional to the length of $A_u$, by 
traversing it in reverse order. Thus, since $c$ is a constant, the entire 
data structure has size $O(n \log n)$ and can be constructed in 
$O(n \log n)$ time. 

Consider a query point $q=(q_x,q_y)$. Observe that a point $p=(p_x,p_y)$ 
is in $\VD_q$ if and only if $p_x \geq q_x$ and $p_y-p_x \geq q_y-q_x$. 
Using the query algorithm for range trees, with query range 
$[q_x,\infty) \times [q_y-q_x,\infty)$, we compute, in $O(\log n)$ total 
time, the sequence $u_i$, $i = 1,2,\ldots,k=O(\log n)$, of nodes in $T$, 
and the sequences of indices $\alpha_i$ and $\beta_i$, $i=1,2,\ldots,k$. 
Observe that each $\beta_i$ is the largest index in the array $A_{u_i}$. 
For each $i$, the array entry $A_{u_i}[\alpha_i]$ stores the $c$ smallest 
weights in the suffix of $A_{u_i}$ that starts at position $\alpha_i$. 
Thus, in $O(c \log n) = O(\log n)$ additional time, we obtain the $c$ 
lowest points of $\VD_q \cap S$. 
\end{proof} 

We conclude this section, by recalling the \emph{Yao-graph}, as 
introduced in Yao~\cite{y-82}. For any point $p=(p_x,p_y)$ in the 
plane, define its four \emph{quadrants} to be the regions (refer to the 
left part of Figure~\ref{figyao})  
\begin{eqnarray*} 
  Q_1(p) & = & [p_x,\infty) \times [p_y,\infty) , \\
  Q_2(p) & = & (-\infty,p_x] \times [p_y,\infty) , \\
  Q_3(p) & = & (-\infty,p_x] \times (-\infty,p_y] , \\
  Q_4(p) & = & [p_x,\infty) \times (-\infty,p_y] . 
\end{eqnarray*} 

\begin{figure}
\begin{center}
\includegraphics[scale=0.65]{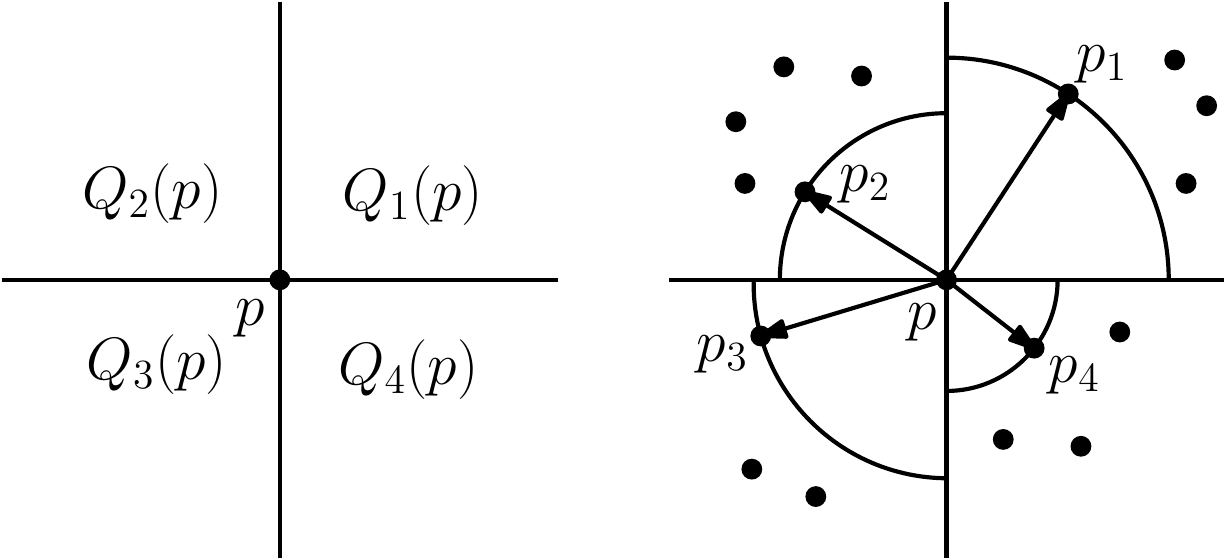}
\end{center}
\caption{On the left, the four quadrants of $p$ are shown. On the right, 
the four out-going edges of $p$ in $\Yao(S)$ are shown.}
\label{figyao}
\end{figure}

Let $S$ be a set of $n$ points in the plane. The Yao-graph is the 
directed graph $\Yao(S)$ with vertex set $S$, whose edge set is obtained 
in the following way (refer to the right part of Figure~\ref{figyao}). 
For any point $p$ in $S$ and any $k=1,2,3,4$ such that 
$Q_k(p) \cap (S \setminus \{p\}) \neq \emptyset$, let $p_k$ be the point 
in $Q_k(p) \cap (S \setminus \{p\})$ whose distance to $p$ is minimum. 
Then, the directed edge $(p,p_k)$ is added to the edge set of $\Yao(S)$. 

The following result is due to Chang et al.~\cite{cht-90}: 

\begin{lemma} 
\label{lemmaYao} 
Let $S$ be a set of $n$ points in the plane. The graph $\Yao(S)$ can 
be computed in $O(n \log n)$ time using $O(n)$ space.  
\end{lemma} 

We partition the edges of the Yao-graph into four subgraphs, based on 
the orientation of these edges: For each $k=1,2,3,4$, $\Yao_k(S)$ 
denotes the directed graph with vertex set $S$ that consists of all edges 
$(p,q)$ in $\Yao(S)$ for which the point $q$ is in the quadrant $Q_k(p)$. 
Each point of $S$ has out-degree zero or one in $\Yao_k(S)$.

\section{The Data Structure for Range Closest Pair Queries} 
\label{secmain} 

In this section, we will present the proof of Theorem~\ref{thm1}. 
Let $S$ be a set of $n$ points in the plane. Our data structure 
consists of the following:  

\begin{enumerate} 
\item $\DS_{\RC}(S)$: This is the data structure of Lemma~\ref{lemma2} 
      that stores the point set $S$ and supports range reporting and 
      counting queries with rectangles.
\item $\DS_{\mbox{\tiny{$\nearrow$}}}(S)$: This is the data structure of 
      Lemma~\ref{lemma4}, with $c=5$, that stores the point set $S$ and 
      returns the smallest square whose bottom-left corner is at a query 
      point and that contains at least $5$ points of $S$ . 
\item $\DS_{\mbox{\tiny{$\nwarrow$}}}(S)$, 
      $\DS_{\mbox{\tiny{$\swarrow$}}}(S)$, 
      $\DS_{\mbox{\tiny{$\searrow$}}}(S)$: 
      These are the three variants of 
      $\DS_{\mbox{\tiny{$\nearrow$}}}(S)$, where bottom-left is replaced 
      by bottom-right, top-right, and top-left, respectively. 
\item Use Lemma~\ref{lemmaYao} to compute the four directed graphs 
      $\Yao_k(S)$, for $k=1,2,3,4$. 
\item For each $k=1,2,3,4$, do the following: Let $S_k$ be the set 
      of all points in $S$ that have out-degree one in the graph 
      $\Yao_k(S)$. Give each point $p$ in $S_k$ a weight which is 
      equal to the length of its out-going edge in $\Yao_k(S)$. 
      Construct the data structure of Lemma~\ref{lemma23} for the 
      weighted point set $S_k$ that reports the minimum weight of any 
      point of $S_k$ inside a query rectangle. We denote this data 
      structure by $\DS_{\MW,k}(S_k)$. 
\end{enumerate} 

It follows from the results in Section~\ref{secprelim} that the total 
preprocessing time is $O(n \log n)$ and the entire data structure has 
size $O(n \log n)$.  

Let $R = [a_x,b_x] \times [a_y,b_y]$ be a query rectangle. Below, we 
will present the algorithm that uses our data structure to compute the 
closest-pair distance in the set $R \cap S$. The algorithm can easily 
be extended such that it also reports the actual closest pair. 

We assume, without loss of generality, that the horizontal side length 
of $R$ is at least the vertical side length. Let $\ell$ be the length of 
the vertical sides of $R$, and let $f$ be the aspect ratio of $R$. 
Thus, the horizontal sides of $R$ have length $f \ell$. 
Refer to Figure~\ref{figqueryR}. 

\begin{figure}
\begin{center}
\includegraphics[scale=0.65]{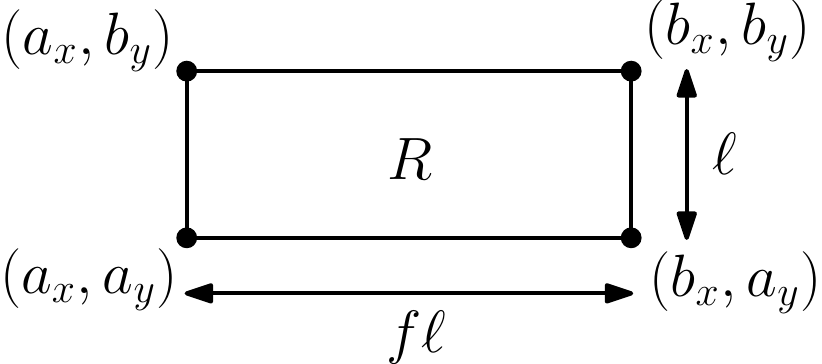}
\end{center}
\caption{The query rectangle $R$.}
\label{figqueryR}
\end{figure}

\begin{figure}
\begin{center}
\includegraphics[scale=0.65]{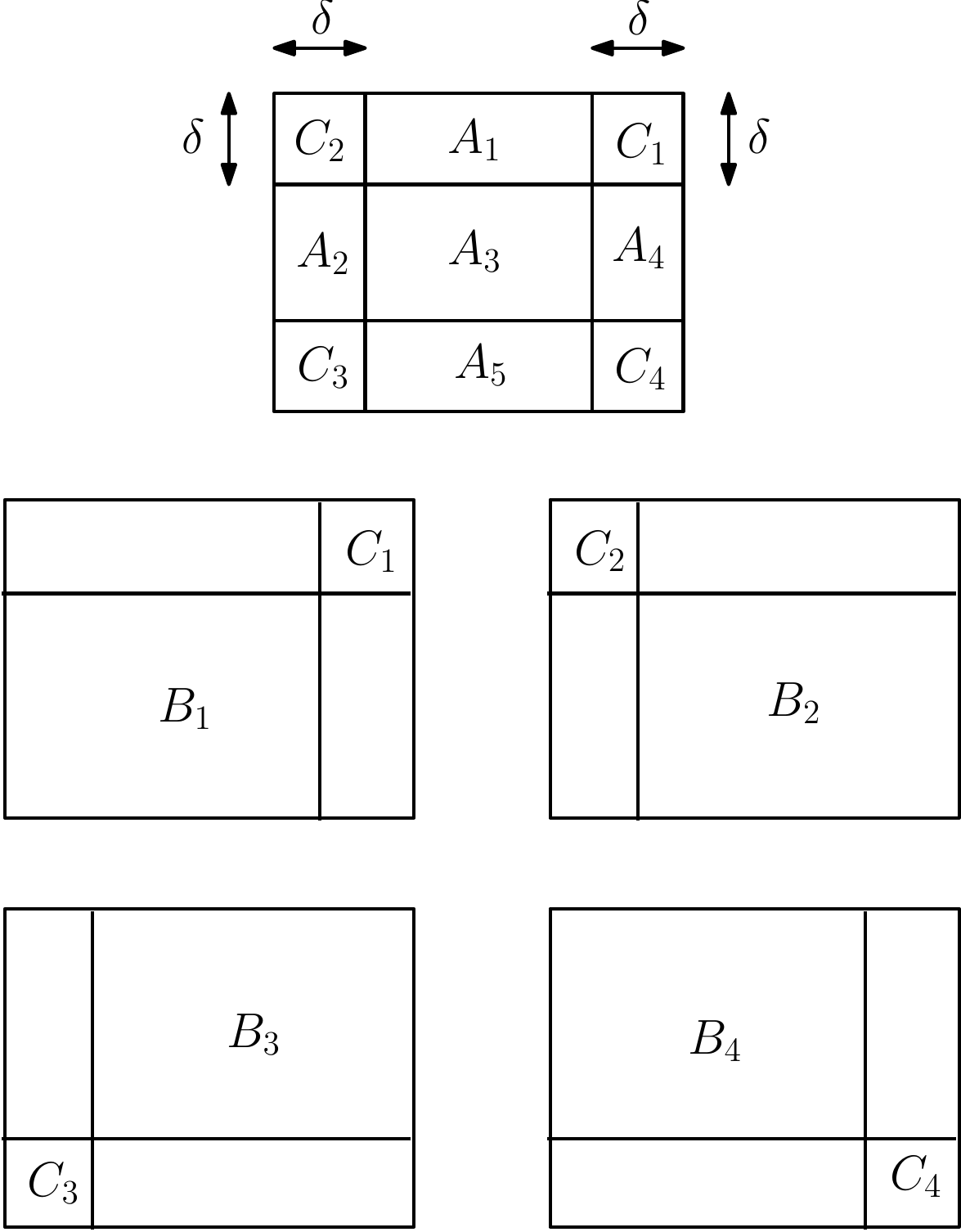}
\end{center}
\caption{Illustrating the partition of the query rectangle $R$.}
\label{fignine}
\end{figure}

Before we present the query algorithm, we introduce some notation. 
Let $\delta$ be a real number with $0 < \delta \leq \ell/2$. 
We denote by $C_1$, $C_2$, $C_3$, and $C_4$ the squares with sides of 
length $\delta$ that are anchored at the four corners of the query 
rectangle $R$, as indicated in Figure~\ref{fignine}. By drawing 
horizontal and vertical lines through the horizontal and vertical sides 
of these squares, the part of the rectangle $R$ without these four 
squares is partitioned into five rectangles. We denote these rectangles 
by $A_1, A_2, \ldots, A_5$, as indicated in Figure~\ref{fignine}. 
Next we define the following four rectangles (again, refer to 
Figure~\ref{fignine}): 
\begin{eqnarray*} 
  B_1 & = & C_3 \cup A_2 \cup A_3 \cup A_5 , \\ 
  B_2 & = & C_4 \cup A_3 \cup A_4 \cup A_5 , \\
  B_3 & = & C_1 \cup A_1 \cup A_3 \cup A_4 , \\
  B_4 & = & C_2 \cup A_1 \cup A_2 \cup A_3 . 
\end{eqnarray*} 

\begin{lemma} 
\label{obs1} 
Let $p$ and $q$ be two points in the plane such that $p$ is to the left
of $q$ and $|pq| < \delta$. 
\begin{enumerate} 
\item If $q \in Q_1(p)$ and $p \in B_1$, then $q \in R$. 
\item If $q \in Q_1(p)$ and $q \in B_3$, then $p \in R$. 
\item If $q \in Q_4(p)$ and $p \in B_4$, then $q \in R$. 
\item If $q \in Q_4(p)$ and $q \in B_2$, then $p \in R$.
\item If both $p$ and $q$ are in $R$ and $q \in Q_1(p)$, then at 
least one of the following holds: (i) $p \in B_1$, (ii) $q \in B_3$, 
(iii) both $p$ and $q$ are in $C_2$, (iv) both $p$ and $q$ are in $C_4$. 
\item If both $p$ and $q$ are in $R$ and $q \in Q_4(p)$, then at 
least one of the following holds: (i) $p \in B_4$, (ii) $q \in B_2$, 
(iii) both $p$ and $q$ are in $C_1$, (iv) both $p$ and $q$ are in $C_3$. 
\end{enumerate} 
\end{lemma} 
\begin{proof}
By symmetry, it suffices to prove the first and fifth claims. To prove 
the first claim, assume that $q \in Q_1(p)$ and $p \in B_1$. Recall that 
the square $C_1$ has sides of length $\delta$. Since $|pq| < \delta$, 
both the horizontal and vertical distances between $p$ and $q$ are 
less than $\delta$. This implies that $q$ is contained in $R$.

To prove the fifth claim, assume that both $p$ and $q$ are in $R$ and 
$q \in Q_1(p)$. If $p \in B_1$, then (i) holds. Assume that 
$p \not\in B_1$. Then $p \in C_2 \cup A_1 \cup C_1 \cup A_4 \cup C_4$. 
If $p \in C_2$, then $q \in C_2 \cup A_1 \cup C_1 \subseteq C_2 \cup B_3$
and, thus, (ii) or (iii) holds. 
If $p \in A_1$, then $q \in A_1 \cup C_1 \subseteq B_3$ and, thus, 
(ii) holds. 
If $p \in C_1$, then $q \in C_1 \subseteq B_3$ and, thus (ii) holds. 
If $p \in A_4$, then $q \in A_4 \cup C_1 \subseteq B_3$ and, thus, 
(ii) holds. Finally, if $p \in C_4$, then 
$q \in C_4 \cup A_4 \cup C_1 \subseteq C_4 \cup B_3$ and, thus (ii) 
or (iv) holds.  
\end{proof} 

We are now ready to present the algorithm that computes the closest-pair 
distance in $R \cap S$: 

\vspace{0.5em} 

\noindent
{\bf Step 1:} Use the data structure $\DS_{\RC}(S)$ to compute 
$|R \cap S|$. 
\begin{itemize} 
\item If $|R \cap S| \leq 4 \lceil 4f \rceil$, then use the data structure 
$\DS_{\RC}(S)$ to compute the elements of the set $R \cap S$, and 
use the algorithm of Lemma~\ref{lemma3} to compute the closest-pair 
distance in $R \cap S$. Return this closest-pair distance and terminate 
the query algorithm. 
\item If $|R \cap S| > 4 \lceil 4f \rceil$, proceed with Step~2. 
\end{itemize} 

\vspace{0.5em} 

\noindent
{\bf Step 2:} Use the data structures 
$\DS_{\mbox{\tiny{$\nearrow$}}}(S)$,
$\DS_{\mbox{\tiny{$\nwarrow$}}}(S)$, 
$\DS_{\mbox{\tiny{$\swarrow$}}}(S)$, and 
$\DS_{\mbox{\tiny{$\searrow$}}}(S)$ to compute the following squares: 
\begin{enumerate} 
\item The smallest square with bottom-left corner at $(a_x,a_y)$ that 
      contains at least $5$ points of $S$. 
\item The smallest square with bottom-right corner at $(b_x,a_y)$ that 
      contains at least $5$ points of $S$. 
\item The smallest square with top-right corner at $(b_x,b_y)$ that 
      contains at least $5$ points of~$S$. 
\item The smallest square with top-left corner at $(a_x,b_y)$ that 
      contains at least $5$ points of $S$. 
\end{enumerate} 
Let $\ell'$ be the side length of the smallest of these four squares. 
If $\ell' > \ell/2$, then let $\delta = \ell/2$. Otherwise, let 
$\delta = \ell'$. Observe that the value of $\delta$ determines the 
squares $C_k$ and the rectangles $B_k$, $1 \leq k \leq 4$. 

\vspace{0.5em} 

\noindent
{\bf Step 3:} For each $k=1,2,3,4$, do the following: Use the data 
structure $\DS_{\MW,k}(S_k)$ to compute the minimum weight of any 
point in $B_k \cap S_k$. If this minimum weight is less than~$\delta$,  
then let $w_k$ be this minimum weight. Otherwise, let $w_k = \infty$. 
 
Compute the value $\delta_1 = \min \{ w_k : 1 \leq k \leq 4 \}$.   

\vspace{0.5em} 

\noindent
{\bf Step 4:} For each $k=1,2,3,4$, do the following: 
Use the data structure $\DS_{\RC}(S)$ to compute the elements 
of the point set $C_k \cap S$. Then use a brute-force algorithm to 
compute the closest pair among these elements. Let $w'_k$ denote the 
closest-pair distance. (In case $|C_k \cap S| \leq 1$, we have 
$w'_k = \infty$.)  

Compute the value $\delta_2 = \min \{ w'_k : 1 \leq k \leq 4 \}$.   

\vspace{0.5em} 

\noindent
{\bf Step 5:} Return the minimum value among $\delta_1$ and $\delta_2$.   

\vspace{0.5em} 

Before we prove the correctness of this algorithm, we analyze its 
running time. By Lemmas~\ref{lemma3} and~\ref{lemma2}, Step~1 takes 
$O(\log n + f \log f)$ time. By Lemma~\ref{lemma4}, Step~2 takes 
$O(\log n)$ time, whereas Step~3 takes $O(\log n)$ time by 
Lemma~\ref{lemma23}. Since each of the squares $C_k$, $1 \leq k \leq 4$, 
contains $O(1)$ points of $S$, it follows from Lemma~\ref{lemma2} that 
Step~4 takes $O(\log n)$ time. Thus, the overall query time is 
$O(\log n + f \log f)$. 

\subsection{Correctness of the Query Algorithm} 
\label{seccorrect} 
To complete the proof of Theorem~\ref{thm1}, it remains to prove the 
correctness of the query algorithm. 
If $|R \cap S| \leq 4 \lceil 4f \rceil$, then the query algorithm 
returns the correct closest-pair distance in the set $R \cap S$. 

Assume that $|R \cap S| > 4 \lceil 4f \rceil$. It follows from the first 
four claims in Lemma~\ref{obs1} that each finite value $w_k$ 
computed in Step~3 is a distance between two distinct points of 
$R \cap S$. Since each square $C_k$ is contained in $R$, the same 
is true for each finite value $w'_k$ computed in Step~4. Thus, the value 
returned in Step~5 cannot be smaller than $\CPD(R \cap S)$. 

By Lemma~\ref{lemma1}, we have $\CPD(R \cap S) < \ell/2$. 
Consider the value $\delta$ that is computed in Step~2. Then,
$0 < \delta \leq \ell/2$ and, again by Lemma~\ref{lemma1}, we have 
$\CPD(R \cap S) < \delta$. 

Let $p,q$ be the closest pair in the set $R \cap S$ and assume, without 
loss of generality, that $p$ is to the left of $q$ and $q \in Q_1(p)$. 
We will prove that either the value $\delta_1$ computed in Step~3 is 
equal to $|pq|$ or the value $\delta_2$ computed in Step~4 is equal to 
$|pq|$. By the fifth claim in Lemma~\ref{obs1}, there are four 
possible cases. 

First assume that $p \in B_1$. Consider the disk of radius $\delta$ that 
is centered at $p$. The quarter of this disk that is inside $Q_1(p)$ is 
completely contained in the rectangle $R$. Since $p,q$ is the closest 
pair in $R \cap S$, this quarter disk does not contain any point of $S$ in 
its interior. Therefore, $(p,q)$ is an edge in $\Yao_1(S)$. It follows
that the value $w_1$ computed in Step~3 is equal to $|pq|$. This, in 
turn, implies that the value $\delta_1$ computed in Step~3 is equal 
to $|pq|$. 

The second case is when $q \in B_3$. By a symmetric argument, $(q,p)$ 
is an edge in $\Yao_3(S)$. Thus, the value $w_3$ computed in Step~3 is 
equal to $|pq|$, implying that $\delta_1 = |pq|$. 

The two remaining cases are when both $p$ and $q$ are in $C_2$ or in
$C_4$. In the former case, the value $w'_2$ computed in Step~4 is equal 
to $|pq|$, whereas in the latter case, $w'_4 = |pq|$. In either case, 
the value $\delta_2$ computed in Step~4 is equal to $|pq|$. 
This concludes the correctness proof of our query algorithm. 

\section{Concluding Remarks} 
\label{seccon}

\subsection{Solutions Based on Candidate Pairs}  
All previous work on the range closest pair problem heavily depends on 
the notion of a \emph{candidate pair}. Let $\mathcal{F}$ be a family of 
regions in the plane. For a given set $S$ of $n$ points in the plane, 
a pair $p,q$ of distinct points in $S$ is called a candidate pair
with respect to $\mathcal{F}$, if there exists a region $R$ in
$\mathcal{F}$ such that $p,q$ is the closest pair in $R \cap S$.

For example, Gupta et al.~\cite{gjks-14} have shown that for the family 
of quadrants $Q_1(a)$, $a \in \IR^2$, two candidate pairs cannot cross. 
Thus, the number of candidate pairs for this family is $O(n)$. 
For the family of vertical slabs, Sharathkumar and Gupta~\cite{sg-07} 
have shown that the number of candidate pairs is $O(n \log n)$. 

Consider the set of all candidate pairs for a given family $\mathcal{F}$ 
of regions. These pairs define a graph with vertex set $S$ that has 
one edge for each candidate pair. We give each such edge a weight which 
is the Euclidean distance between its two vertices. Then a range closest 
pair query for a given region $R$ in $\mathcal{F}$ reduces to 
determining the shortest edge that is completely contained inside $R$. 
Obviously, the amount of space used by this approach is at least 
the number of candidate pairs. The following lemma shows that this 
approach does not lead to a space-efficient data structure for fat 
rectangles.   

\begin{lemma} 
Let $f>1$ be a real number and let $\mathcal{F}_f$ be the family 
consisting of all rectangles in the plane having aspect ratio at most $f$. 
There exists a set $S$ of $n$ points in the plane, for which the number 
of candidate pairs with respect to $\mathcal{F}_f$ is $\Omega(n^2)$. 
\end{lemma} 
\begin{proof} 
Consider a circle centered at the origin. Let $a$ and $b$ be the two 
intersection points of this circle with the line $y=x$. Let $S_a$ be a 
set of $n/2$ points on this circle that are very close to $a$, and let 
$S_b$ be a set of $n/2$ points on this circle that are very close to $b$. 
Let $S = S_a \cup S_b$. For any point $p$ in $S_a$ and any point $q$ in 
$S_b$, the rectangle $Q$ with corners $p$ and $q$ has aspect ratio at 
most $f$ and does not contain any other point of~$S$. Thus, $p,q$ is a 
candidate pair. As a result, for the family $\mathcal{F}_f$, the number 
of candidate pairs in $S$ is at least $(n/2)^4$. 
\end{proof} 

Next we show that for the family of axes-parallel squares (i.e., 
the family $\mathcal{F}_f$ with $f=1$), the number of candidate pairs 
is $O(n)$. 

\begin{lemma}  \label{lemmafinal} 
Let $\mathcal{F}$ be the family consisting of all axes-parallel squares 
in the plane, and let $S$ be a set of $n$ points in the plane. The 
number of candidate pairs in $S$ with respect to $\mathcal{F}$ is 
$O(n)$. 
\end{lemma} 
\begin{proof} 
For any integer $k \geq 0$, the \emph{order}-$k$ 
$L_{\infty}$-\emph{Delaunay graph} $\Del(S,k)$ is the graph with vertex 
set $S$ in which any two distinct points $p$ and $q$ form an edge if 
their exists an axes-parallel square that has $p$ and $q$ on its 
boundary and contains at most $k$ points of $S$ in its interior.

Let $p,q$ be an arbitrary candidate pair and let $R$ be a square such 
that $p,q$ is the closest pair in $R \cap S$. We first prove that there 
exists a square $R'$ in $R$ such that $p$ and $q$ are on opposite sides 
of $R'$.

To prove this claim, let $\ell$ be the side length of $R$. We may assume, 
without loss of generality, that $q \in Q_1(p)$ and the horizontal 
distance $h$ between $p$ and $q$ is at most their vertical distance $v$.  
We take for $R'$ any square with sides of length $v$ that contains $p$ 
on the bottom side, $q$ on the top side, and that is contained in $R$. 
 
Obviously, $p,q$ is the closest pair in $R' \cap S$, implying that 
$\CPD(R' \cap S) \geq \ell'$, where $\ell'$ is the side length of $R'$. 
It then follows from the second claim in Lemma~\ref{lemma1} that $R'$ 
contains at most four points of $S$. Since $p$ and $q$ are two of 
these points, $R'$ contains at most two points of $S$ in its interior. 
Therefore, $pq$ is an edge in $\Del(S,2)$.

Bose et al.~\cite[Corollary~$12$]{bch-10} have shown that the order-$k$ 
Delaunay graph can be partitioned into at most $18 k^2$ graphs, each of
which is plane. This implies that the order-$k$ Delaunay graph has at 
most $O(k^2 n)$ edges. Even though they prove this result for the 
Euclidean metric, their arguments are valid for the $L_{\infty}$-metric 
as well. Thus, the number of edges in $\Del(S,2)$ is $O(n)$. This implies 
that the number of candidate pairs for the family of squares is $O(n)$. 
\end{proof} 

Lemma~\ref{lemmafinal} implies that the approach based on candidate 
pairs can be used to obtain a space-efficient solution for closest-pair 
queries with squares. The main drawbacks are that it is not clear how to 
compute all candidate pairs in $O(n \log n)$ time, and how to compute the 
shortest candidate pair inside a query square in $O(\log n)$ time.   

\subsection{Open Problems} 
The correctness of our query algorithm heavily uses the fact that the 
closest-pair distance is less than half of the shortest side length 
$\ell$ of the query rectangle, in case it contains $\Omega(f)$ points. 
Consider the value $\ell'$ that is computed in Step~2 of the query 
algorithm. If $\ell' \leq \ell/2$, then Step~1 is not necessary, and 
the query algorithm takes $O(\log n)$ time, even if the aspect ratio 
of the query rectangle is very large. We leave as an open problem 
to design a data structure of size $O(n \log n)$ and query time 
$O(\log n)$ that works for any query rectangle. Also, we leave open 
the problem of improving the space bound for fat rectangles to $O(n)$. 

For the range closest pair problem in $\IR^d$, with $d \geq 3$, no 
non-trivial results are known. In particular, if $d=3$, it is not known 
if queries can be answered in polylogarithmic time using a data 
structure whose size is close to linear.

\section*{Acknowledgements} 
This work was initiated at the 21st Korean Workshop on Computational 
Geometry, held in Rogla, Slovenia, in June 2018. 
The authors thank the other workshop participants for their helpful 
comments.

\bibliographystyle{plain}
\bibliography{ClosestPairFat}

\end{document}